 			\def\version{20 December, 2011}                        % 
\documentclass[reqno,11pt]{amsart} 

\usepackage{amsmath,amssymb,amsthm,amsfonts}
\usepackage{color}
\usepackage{hyperref}
\usepackage[active]{srcltx} 
\newtheorem{theorem}{Theorem}[section] 
\newtheorem{lemma}[theorem]{Lemma} 
\newtheorem{prop}[theorem] {Proposition}

\newtheorem{hypo}{Hypothesis}

\newcommand{\N}{\mathbb{N}}
\newcommand{\R}{\mathbb{R}}

\newcommand{\dd}{\mathrm{d}} %integration d
\newcommand{\eps}{\epsilon}

\renewcommand{\eps}{\varepsilon}
\newcommand{\smfrac}[2]{\textstyle{\frac {#1}{#2}}}

\newcommand{\vect}[1]{\boldsymbol{#1}}

\DeclareMathOperator{\supp}{supp}

\newcommand{\sat}{\mathrm{sat}}
\newcommand{\ideal}{{\rm ideal}}
\newcommand{\cl}{{\rm cl}} 
\renewcommand{\sat}{\mathrm{sat}}
\newcommand{\Ncal}   {{\mathcal N }} 

\newcommand{\ssup}[1] {{\scriptscriptstyle{({#1}})}} 
\def\1{{\mathchoice {1\mskip-4mu\mathrm l}      % Blackboard bold 1 
{1\mskip-4mu\mathrm l} 
{1\mskip-4.5mu\mathrm l} {1\mskip-5mu\mathrm l}}} 
\newcommand{\e}   {{\operatorname e }}

\begin{document} 

%opening
\title[Ideal mixture approximation of cluster size distributions at low density]{Ideal mixture approximation\\ of cluster size distributions at low density}

\author[Sabine Jansen and Wolfgang K{\"o}nig]{} 

\maketitle

\thispagestyle{empty} 
\vspace{0.2cm} 
 
\centerline 
{\sc Sabine Jansen$ $\footnote{Weierstrass Institute Berlin, Mohrenstr.~39, 10117 Berlin, {\tt jansen@wias-berlin.de, koenig@wias-berlin.de}}, Wolfgang K{\"o}nig$ $\footnotemark[1]$^,$\footnote{Technische Universit\"at Berlin, Str.~des 17.~Juni 136, 10623 Berlin}}

\medskip

\centerline{WIAS Berlin and TU Berlin}
\vspace{0.2cm}

\centerline{\small(\version)}

\vspace{.2cm}

\begin{quote}	 
	We consider an interacting particle system in continuous configuration space.
	 The pair interaction has an attractive part. 
	We show that, at low density, the system behaves approximately like an ideal 
	mixture of clusters (droplets): we  prove rigorous bounds (a) for the constrained free energy associated with a given cluster size distribution, considered as an order parameter, 
	(b) for the free energy, obtained by minimising over the order parameter, 
	and (c) for the minimising cluster size distributions. It is known that, under suitable assumptions, 
	the ideal mixture has a transition from a gas phase to a condensed phase as the density is varied; our bounds hold both in the gas phase and in the coexistence region of the ideal mixture.
	The present paper improves our earlier results by taking into account the mixing entropy. 
\end{quote}

\bigskip\noindent
{\it MSC 2010.} Primary \texttt{82B21}, Secondary \texttt{60F10; 82B31; 82B05.}

\medskip\noindent
{\it Keywords and phrases.} Classical particle system, canonical ensemble, equilibrium statistical mechanics, dilute system, large deviations.

\setcounter{section}{0} 
\setcounter{tocdepth}{1}
%\tableofcontents

\section{Introduction and main results}

\noindent The idea to treat an interacting particle gas as an approximately ideal mixture of droplets (clusters) is classical and of wide-spread use in statistical mechanics, thermodynamics, and physical chemistry. It goes sometimes under the name of \emph{Frenkel-Band theory of association equilibrium}, see the textbook~\cite{hillbook}. From a more mathematical perspective, this droplet picture has been used in investigations of percolation properties for lattice and continuum systems~\cite{penrose-lebowitz,muermann,ghm01}. 

This article's main concern is to make this droplet picture rigorous. To the best of our knowledge, all existing results work in the grand-canonical ensemble at sufficiently negative values of the chemical potential, for which one expects that all clusters are finite. 
In contrast, we work in the canonical ensemble. This allows us to give results that hold also in a regime where there might be infinite clusters.

In this work, we consider the free energies, the constrained free energies with fixed cluster size distributions, and the optimal distribution itself and derive bounds for the deviation from the ideal model. These bounds decay exponentially fast as a function of the inverse temperature at low densities, respectively they decay as a power of the density.

\subsection{The model} 

\noindent  We consider a system $\vect{x}=(x_1,\dots,x_N)$ of $N$ particles in a box $\Lambda=[0,L]^d$ with interaction given by
\begin{equation}\label{Udef}
U_N(\vect{x})=U_N(x_1,\ldots, x_N)= \sum_{1\leq i<j\leq N} v(|x_i-x_j|),
\end{equation}
where $v\colon [0,\infty) \to \R\cup \{\infty\}$ is a Lennard-Jones-type potential. Our precise assumptions, which are the same as in \cite{jkm}, are as follows.

\medskip

\noindent \textbf{Assumption (V).} {\it The function $v\colon [0,\infty) \to \R\cup \{\infty\}$ satisfies the following.
\begin{itemize}
	\item[(1)] $v$ is finite except possibly for a hard core: there is a $r_\mathrm{hc}\geq 0$ 
		such that 
		$$ \forall r \leq  r_\mathrm{hc}:\ v(r)= \infty, \qquad 
			\forall r > r_\mathrm{hc}:\ v(r) \in \R.
		$$
	\item[(2)] $v$ is stable.
	\item[(3)] The support of $v$ is compact, more precisely,  $b:=\sup \supp v <\infty$.
	\item[(4)] $v$ has an attractive tail: there is a $\delta>0$ such that $v(r)<0$ for all 
		$r \in (b-\delta,b)$.
	\item[(5)] $v$ is continuous in $[r_\mathrm{hc},\infty)\cap(0,\infty)$.
\end{itemize}
}
\medskip

Throughout this paper, Assumption (V) will be in force without further mentioning. We consider the thermodynamic limit $N,L\to\infty$ such that $|\Lambda|=L^d=N/\rho$ for some particle density $\rho\in(0,\infty)$ at positive and finite inverse temperature $\beta\in(0,\infty)$. The existence of the free energy per unit volume is well-known: there is a $\rho_\mathrm{cp}\in(0,\infty]$, the \emph{close-packing density}, such that for all $\rho\in (0,\rho_\mathrm{cp})$,  the limit  
\begin{equation}\label{FreeEnergy}
f(\beta,\rho):= - \frac{1}{\beta} \lim_{N,L \to \infty,L^d=N/\rho} 
\frac{1}{|\Lambda|} \log \Big(\frac{1}{N!} \int_{\Lambda^N}  \e^{-\beta U_N(\vect{x})} \,\dd \vect{x}\Big),\qquad\beta \in(0,\infty)
\end{equation}
exists in $\R$.  In the following, we always assume that $\beta \in(0,\infty)$ and $\rho\in (0,\rho_\mathrm{cp})$.

Our main concern is the cluster size distribution that is induced by the Gibbs measure as a random sequence. Fix $R\in(b,\infty)$. For a given configuration $\vect{x}=(x_1,\ldots,x_N) \in \Lambda^N$, draw an edge between two points $x_i$ and $x_j$ if their distance $|x_i - x_j|$ is $\leq R$. In this way, the configuration splits into connected components, which we call \emph{clusters}. Let 
$N_k(\vect{x})$ be the number of $k$-clusters, i.e., components with $k$ particles,  and let
\begin{equation*}
\rho_{k,\Lambda}(\vect{x}):= \frac{N_k(\vect{x})}{|\Lambda|}
\end{equation*}
be the number of $k$-clusters per unit volume. We consider the {\em cluster size distribution} $\vect{\rho}_\Lambda=( \rho_{k,\Lambda})_{k\in \N}$ as a random variable in the set of all sequences $\vect{\rho}=(\rho_k)_{k\in\N}\in[0,\infty)^\N$. One of the main results of \cite{jkm} is the existence of a {\em rate function} $f(\beta,\rho,\cdot)\colon [0,\infty)^\N \to \R \cup \{\infty\}$ such that, in the above thermodynamic limit, 
\begin{equation}\label{LDP}
-\frac1\beta \frac 1{ |\Lambda|}\log\Big(\frac{1}{N!} \int_{\Lambda^N} \e^{-\beta U_N(\vect{x})} 
		\1\bigl\{\vect{\rho}_\Lambda(\vect{x})\in\cdot \bigr\}\, \dd \vect{x} \Big)
		\Longrightarrow \inf_{\vect{\rho}\in\cdot}f(\beta,\rho,\vect{\rho} ),
\end{equation}
in the sense of a large-deviation principle. That is, the limit in \eqref{LDP} holds in the weak sense, and the level sets of the rate function $f(\beta,\rho,\cdot)$ are compact; in particular, $f(\beta,\rho,\cdot)$ is lower semi-continuous.  Furthermore,
$f(\beta,\rho,\cdot)$ is convex. Moreover, if $f(\beta,\rho,\vect{\rho})$ is finite, then necessarily $\sum_{k=1}^\infty k \rho_k \leq \rho$. At low density $\rho$, the converse is also true, i.e., $f(\beta,\rho,\vect{\rho})<\infty$ for any $\vect{\rho}$ such that $\sum_{k=1}^\infty k \rho_k \leq \rho$. 
The relation with the free energy  given in \eqref{FreeEnergy} is 
\begin{equation*}
	f(\beta,\rho) = \min \Bigl \lbrace f\bigl(\beta,\rho,\vect{\rho} \bigr) 
		\,\Big|\, \vect{\rho} \in [0,\infty)^\N,\ \sum_{k\in\N} k \rho_k \leq \rho \Bigr\rbrace.
\end{equation*}
Hence, the minimiser(s) of $f(\beta,\rho,\cdot)$ play an important role as the optimal configuration(s) of the system.

\subsection{The ideal-mixture model}

 \noindent We now introduce the main object in terms of which we will approximate the above model. The {\it cluster partition function} of a  $k$-cluster is introduced as
\begin{equation}\label{Zcldef}
   Z_k^\cl(\beta) := \frac{1}{k!} \int_{ (\R^d)^{k-1}} \e^{-\beta U_k(0,x_2,\ldots,x_k)} 
	\1\big\{\{0,x_2,\ldots,x_k\}\ \text{connected}\big\}\,\dd x_2\cdots \dd x_k,
\end{equation}
and the associated cluster free energy per particle is $f_k^\cl(\beta): = - \frac{1}{\beta k } \log Z_k^\cl(\beta)$. It is known that its thermodynamic limit,
\begin{equation}\label{cfe}
f_\infty^\cl(\beta)= \lim_{k\to \infty} f_k^\cl(\beta),
\end{equation}
exists in $\R$. Indeed, this was proved in \cite{DS84}, given that $(f_k^\cl(\beta))_{k\in \N}$ is bounded, and the boundedness of $(f_k^\cl(\beta))_{k\in \N}$ was 
derived in \cite{DS84} in dimension $d=2$ and $d=3$ and in~\cite[Lemma 4.3 and 4.5]{jkm} in all dimensions. 

For $\vect{\rho}=(\rho_k)_{k\in\N}$, let 
\begin{equation}\label{fideal}
	f^\ideal(\beta,\rho,\vect{\rho}) :=  \sum_{k\in \N} k \rho_k f_k^\cl(\beta) + \Bigl(\rho- \sum_{k\in \N} k \rho_k\Bigr) f_\infty^\cl(\beta) 
		+ \frac 1\beta \sum_{k\in \N} \rho_k (\log \rho_k - 1). 
\end{equation}
This rate function describes the large deviations of the cluster size distribution in an idealised model that neglects the excluded-volume effect:  the first term describes the 
internal free energy coming from  the clusters of finite size, the second term the analogous contribution from clusters of infinite size, and the last term describes the entropy of placing all these clusters into the volume, not taking care of being separated from each other. See Section~\ref{sect-discussion} for an integer partition model that has $f^\ideal(\beta,\rho,\cdot)$ as a rate function.

We are going to compare $f(\beta,\rho,\vect{\rho})$ with $f^\ideal(\beta,\rho,\vect{\rho})$ for small densities $\rho$ and low temperatures $1/\beta$. That these two should be close to each other is intuitively clear, since low temperature should ensure that clusters assume a compact shape, and low density should give enough space to place the clusters at positive mutual distance. 
 The main purpose of  this paper is to make this reasoning rigorous.

\subsection{Our hypotheses}\label{sect-hypo}

\noindent We need further assumptions about ground states and the cluster partition functions. Roughly speaking, we need to assume some H\"older continuity of the energy $U_k(\cdot)$ close to the ground states and that the relevant clusters at zero and low temperature, respectively, have a compact shape, i.e., occupy at most a box with volume of order of the number of particles. These hypotheses are believed to be true for many potentials of the type in Assumption (V), and they can be seen to be satisfied for the ground states. However, for positive temperatures, their rigorous understanding has not yet been completed.

The purpose of the hypothesis of bounded density is the following. As we indicated above, the fundamental idea is to split the configuration into its clusters and to collect the internal free energies of all the clusters. However, one also needs to describe the entropy of a configuration, that is, the combinatorial complexity for the placement of all the clusters into some cube. This task is very hard without further information. In the low-density approximation, we will solve this task by neglecting the excluded-volume effect, which makes it much easier. For this, we need to know that the clusters do not require a large diameter. Our second hypothesis ensure this for the ground states, and the last two hypotheses ensure this for positive low temperature.

First we formulate our hypothesis about uniform H\"older continuity of the energy around the ground states and a strong form of stability. Recall that the pair potential $v$ is called stable if $\frac 1k\inf_{\vect{x}\in(\R^d)^k}U_k(\vect{x})$ is bounded from below in $k$, which means that the ground states do not clump too strongly.

\begin{hypo} \label{ass:hoelder}
There is a $r_\mathrm{min}\in(r_\mathrm{hc},\infty)$ such that $v$ is uniformly H{\"o}lder continuous in $[r_\mathrm{min},\infty)$ and, for all $k \in \N$, every minimiser $(x_1,\ldots,x_k) \in (\R^d)^k$ of the energy $U_k$ has interparticle distance lower bounded as $|x_i- x_j| \geq r_\mathrm{min}$ for any $i\neq j$.
\end{hypo}

This hypothesis can be seen to be satisfied under some mild additional assumptions on $v$ relating the negative part of $v$ to its behaviour at zero; see \cite[Proof of Lemma~3.1]{CKMS10} or \cite[Lemma 2.2]{Th06}. The H{\"o}lder continuity allows us to give low-temperature estimates of the form 
\begin{equation*}
	-\frac{1}{\beta k}\log \frac{1}{k!} \int_{(\R^d)^k} \e^{-\beta U(\vect{x})}\, \dd \vect{x} 
	= \frac{1}{k} \inf_{\vect{x}\in (\R^d)^k} U(\vect{x}) + O\Big(\frac{\log \beta}{\beta}\Big) \qquad \text{as}\ \beta \to \infty,
\end{equation*}
uniformly in $k\in \N$.

Our next hypothesis says that the minimising configurations (the ground states) do not occupy more space than a box with volume of order of the number of particles.

\begin{hypo}[Ground states have a compact shape] \label{ass:maxdist}
 	There is a constant $c>0$ such that for all $k\in \N$
	 every minimiser $(x_1,\ldots,x_k) \in (\R^d)^k$ of the energy $U_k$
	has interparticle distance upper bounded by $|x_i - x_j| \leq c k ^{1/d}$ for any $i\neq j$.
\end{hypo}

This hypothesis is known to be satisfied for some classes of potentials having a large intersection with those satisfying Assumption (V), however, only in dimension one and two. See~\cite{R81} and~\cite{yeung}.

Now we proceed with two more restrictive hypotheses, whose validity has actually not been clarified for all interesting potentials of the type that we consider. They concern, at positive  sufficiently low temperature, the diameter of the relevant clusters. 

An important object is the internal cluster energy coming from a box of volume $a^d$:
\begin{equation}\label{Zcladef}
   Z_k^{\cl,a}(\beta) := \frac{1}{k!a^d} \int_{ ([0,a]^d)^{k}} \e^{-\beta U_k(\vect{x})} 
	\1\{\vect{x}\ \text{connected}\}\,\dd x_1\cdots \dd x_k.	
\end{equation}
The reader may  check that $\lim_{a\to \infty} Z_k^{\cl,a}(\beta) = Z_k^\cl(\beta)$ holds for every fixed $k$ and $\beta$.
The corresponding free energy is defined as $f_k^{\cl,a}(\beta): = - \frac{1}{\beta k } \log Z_k^{\cl,a}(\beta)$.
It is tempting to believe (and this is the content of our next hypothesis) that, at least at sufficiently low temperature, a box of volume of order $k$ should capture almost all the internal free energy of a cluster:

\begin{hypo}[Clusters have compact shape] \label{collapse} 
For some $c\in(0,\infty)$ and every sufficiently large $\beta$,
	\begin{equation} \label{eq:collapse}
		\lim_{k\to \infty} \frac{1}{k} \log Z_k^{\cl, ck ^{1/d}}(\beta) 
		= \lim_{k\to \infty} \frac{1}{k} \log Z_k^{\cl}(\beta).
	\end{equation}
\end{hypo}

However, this hypothesis has not even been proved for the relatively simple case of the two-dimensional Ising model, see the discussion in~\cite{DS86}. It is commonly believed that \eqref{eq:collapse} is true for low temperature and wrong for high temperature, since the cluster is believed to assume a tree-like structure and to occupy therefore a much larger portion of space. This phenomenon is often called a \emph{collapse transition}: as the temperature decreases below a critical value, the volume per particle collapses to some finite value.

The following hypothesis is in the spirit of  Hypothesis~\ref{collapse} and goes much beyond it: if, for large $\beta$, the relevant configurations for $f_k^\cl(\beta)$ have a compact shape, then the number of particles that have not the optimal number of neighbours should be of surface order. Therefore the correction to the large-$k$ asymptotics should be of surface order of a ball with volume $\approx k$:

\begin{hypo} \label{secondorder}
For some $C>0$ and all sufficiently large $\beta$,
	\begin{equation}\label{fkbound}
		k f_k^\cl(\beta) - k f_\infty^\cl(\beta) \geq C  k^{1-1/d},\qquad k\in \N.
	\end{equation}
\end{hypo}

To the best of our knowledge, such a deep statement has not been proved for any interesting potential satisfying Assumption (V). Actually for our proofs we only need a lower bound against $C  k^\eps$ for some $\eps>0$ instead of $C  k^{1-1/d}$.

\subsection{Our results}\label{sec-Results}

\noindent Our first main result  applies to all cluster size distributions $\vect{\rho}$, not only minimisers of the rate functions, and is therefore possibly of interest for non-equilibrium thermodynamic models. 

\begin{theorem}[Comparison of $f(\beta,\rho,\cdot)$ and $f^\ideal(\beta,\rho,\cdot)$]  \label{thm:bounds}
Let the pair potential $v$ satisfy Hypotheses \ref{ass:hoelder} and~\ref{ass:maxdist}. Then there are $\overline \rho, \overline\beta$ and $ C>0$ such that for all $\beta \in[\overline\beta,\infty)$, $\rho \in (0,\overline \rho)$, and $K\in \N$ with $K <(\rho/3)^{-1/(d+1)}$,  and  all $\vect{\rho}=(\rho_k)_{k\in\N}\in[0,\infty)^\N$ satisfying $\sum_{k\in\N}k\rho_k\leq \rho$,
\begin{equation} \label{ineq:idealunif}
f^\ideal(\beta,\rho,\vect{\rho})  \leq  f(\beta,\rho,\vect{\rho}) \leq   f^\ideal(\beta,\rho, \vect{\rho})  + \frac{C}{\beta} \eps_K(\beta,\rho,\vect{\rho}),
\end{equation}
where
\begin{equation}\label{epsdef}
\eps_K(\beta,\rho,\vect{\rho}) = \rho_{\leq K} ^{(d+2)/(d+1)} + (\rho - \rho_{\leq K}) \log \beta -  m_{>K} \log m_{>K}  
\end{equation}
and we abbreviated $\rho_{\leq K}:= \sum_{k=1}^K k \rho_k$ and $m_{>K}:= \sum_{k=K+1}^\infty \rho_k$. 
If in addition Hypothesis~\ref{collapse} holds, then in \eqref{epsdef}
we can replace $(\rho-\rho_{\leq K})$ with $\sum_{k=K+1}^\infty k \rho_k$. 
\end{theorem}

Theorem~\ref{thm:bounds} is proved in Section~\ref{sect-proofthmbounds}. Next, we compare the minimimum and the minimisers under the two stronger hypotheses on the compact shape of the relevant clusters at positive temperature and the finite size correction of the cluster free energy. Let 
\begin{equation}\label{fidealdef}
f^\ideal(\beta,\rho): = \inf\Bigl\lbrace f^\ideal(\beta,\rho, \vect{\rho}) \,\Big|\,\vect{\rho}\in[0,\infty)^\N,\sum_{k\in\N} k \rho_k \leq \rho \Bigr \rbrace.
\end{equation}
It is not difficult to see that there is a unique minimiser $\vect{\rho}^\ideal(\beta,\rho)=(\rho^\ideal_k(\beta,\rho))_{k\in \N}$. 
We set $m^\ideal(\beta,\rho):= \sum_{k=1}^\infty \rho_k^\ideal(\beta,\rho)\in[0,\rho]$. 

\begin{theorem} \label{thm:exponential}
Suppose that Hypotheses~\ref{ass:hoelder}, \ref{collapse} and~\ref{secondorder} hold, and assume that $d\geq 2$. Then there are $\overline\beta,\overline \rho, C,C'>0$ such that, for all $\beta \in[\overline\beta,\infty)$ and $\rho \in (0,\overline \rho)$, the following holds.
\begin{enumerate}
\item Free energy: 
\begin{equation} \label{eq:expoa}
0 \leq f(\beta,\rho) - f^\ideal(\beta,\rho)  \leq \frac C \beta m^\ideal(\beta,\rho) \rho^{1/(d+1)}. 
\end{equation} 

\item Let $\vect{\rho}= \vect{\rho}^\ssup{\beta,\rho}=(\rho_k)_{k\in\N}$ be a minimiser of $f(\beta,\rho,\cdot)$, and put $m:= \sum_{k\in \N} \rho_k$. Then 
    \begin{equation}\label{MinEsti}
	\left| \frac{m}{m^\ideal(\beta,\rho)} - 1\right|^2  \leq  C' \rho^{1/(d+1)}\qquad\mbox{and}\qquad 
	 \frac 12 H\Bigl( \frac{ \vect{\rho}}{m} ; \frac{ \vect{\rho}^\ideal(\beta,\rho)}{m^\ideal(\beta,\rho)}\Bigr)  \leq  C' \rho^{1/(d+1)}.
    \end{equation}
\end{enumerate}
\end{theorem}

Here 
\begin{equation*}
	H(\vect{a};\vect{b}) = \sum_{k\in \N}\Bigl(b_k-a_k+  a_k \log \frac{a_k}{b_k} \Bigr)
\end{equation*}
is the relative entropy between two finite measures $\vect{a}$ and $\vect{b}$ on $\N$, 
and we recall Pinsker's inequality: when $\vect{a}$ and $\vect{b}$ are probability measures, 
then $\frac 12 H(\vect{a};\vect{b}) \geq  ||\vect{a}- \vect{b}||_\mathrm{var}^2$. 
The proof of Theorem~\ref{thm:exponential} is in Section~\ref{sect-proofthmexpo}.

If we do not assume that Hypotheses~\ref{collapse} and~\ref{secondorder} are true, 
our rigorous bounds hold in the temperature-density plane only in a region away from the critical line given by $\rho=\exp(-\beta \nu^*)$ with  $\nu^*\in(0,\infty)$ defined as follows. Introduce the ground-state energy,
\begin{equation}\label{Ekdef}
	E_k:= \inf_{\vect{x}\in(\R^d)^k} U_k(\vect{x}).
\end{equation}
Then $e_\infty:= \lim_{k\to\infty}E_k/k$ exists in $(-\infty,0)$, and $\nu^*:=\inf_{k\in\N}(E_k-ke_\infty)$ is positive \cite{CKMS10,jkm}. 

\begin{theorem} \label{thm:ideal}
Let $v$ satisfy Hypotheses~\ref{ass:hoelder} and~\ref{ass:maxdist}. Then, for any $\eps >0$ there are $\beta_\eps, C_\eps,C'_\eps>0$ such that for all $\beta \in[\beta_\eps,\infty)$ and $\rho\in(0,\infty)$ satisfying $- \beta^{-1}\log \rho > \nu^* +\eps$, \eqref{eq:expoa} and~\eqref{MinEsti} hold with $C$ and $C'$ replaced by $C_\eps$ and $C'_\eps$, respectively.\end{theorem}

The proof of Theorem~\ref{thm:ideal} is in Section~\ref{sect-proofthmideal}.

\subsection{Discussion}\label{sect-discussion}

Let us explain in more detail the significance of Theorems~\ref{thm:exponential} 
and~\ref{thm:ideal}, and in which way they improve results of~\cite{jkm}. 

We start by recalling the properties of the idealised problem; see also \cite[Sect.~4]{BCP86}.  As mentioned above, 
for all $\beta$ and $\rho$, $f^\ideal(\beta,\rho,\cdot)$ has a unique minimiser 
$(\rho_k^\ideal(\beta,\rho))_{k\in \N}$, which can be characterised as follows. Let 
\begin{equation}\label{rhosatdef}
	\rho_\sat^\ideal(\beta):= \sum_{k\in \N} k\, \e^{\beta k [f_\infty^\cl(\beta) - 
			f_k^\cl(\beta)]} \in(0,\infty]
\end{equation}
be the saturation density of the ideal mixture. 
If Hypothesis~\ref{secondorder} holds (actually also under much weaker bounds than the one in \eqref{fkbound}), at low temperature, $\rho_\sat^\ideal(\beta)$ is finite. 
In general, however, it can be infinite. 
For $\rho<\rho_\sat^\ideal(\beta)$, let $\mu^\ideal(\beta,\rho)\in(-\infty,f_\infty^\cl(\beta))$ be the unique solution of 
\begin{equation} \label{eq:musol}
	 \sum_{k=1}^\infty k \,\e^{\beta k [\mu^\ideal(\beta,\rho)- f_k^\cl(\beta)]} = \rho,
\end{equation}
and for $\rho \geq \rho_\sat^\ideal(\beta)$, let $\mu^\ideal(\beta,\rho):= f_\infty^\cl(\beta)$.
Then, the minimiser $(\rho_k^\ideal(\beta,\rho))_k$ is given by 
\begin{equation} \label{eq:rhok}
	\rho_k^\ideal(\beta,\rho) = \e^{\beta k [\mu^\ideal(\beta,\rho)- f_k^\cl(\beta)]},  
\end{equation}
and the ideal free energy from~\eqref{fidealdef} is given by 
\begin{equation} \label{eq:fideal}
	f^\ideal(\beta,\rho) = \rho \mu^\ideal(\beta,\rho) - \frac{1}{\beta}  m^\ideal(\beta,\rho). 
\end{equation}
Moreover, $\rho \mapsto f^\ideal(\beta,\rho)$ is analytic and strictly convex in $(0,\rho_\sat^\ideal(\beta))$, and linear with slope $f_\infty^\cl(\beta)$ in $[\rho_\sat^\ideal(\beta),\infty)$. In particular, the ideal mixture undergoes a phase transition as the density is varied if and only if the saturation density of the ideal mixture is finite. The transition is from a gas phase where all particles are in finite-size clusters, to a condensed phase where 
a positive fraction goes into unboundedly  large clusters: for all $\rho<\rho_\sat^\ideal(\beta,\rho)$, 
we have $\sum_{k=1}^\infty k \rho_k^\ideal(\beta,\rho) = \rho$, while for 
$\rho>\rho_\sat^\ideal(\beta)$, $\sum_{k=1}^\infty k \rho_k^\ideal(\beta) = \rho_\sat^\ideal(\beta) <\rho$. 

Armed with this knowledge, we can compare our results with those of~\cite{jkm}. In \cite{jkm}, we approximated the rate function $f(\beta,\rho,\cdot)$, more precisely the function $\vect{q}=(q_k)_ {k\in\N}\mapsto \frac 1\rho f(\beta,\rho,(kq_k)_ {k\in\N}/\rho)$, with 
\begin{equation} \label{eq:ckms}
g_\nu(\vect{q}) = \Big(1- \sum_{k=1}^\infty q_k\Big) e_\infty + \sum_{k=1}^\infty q_k
		 \frac{E_k - \nu}{k}, \qquad \nu:= - \frac1\beta{\log \rho}. 
\end{equation} 
This function is easier to formulate, but involves more approximations, and has some rather unphysical properties. This approximation was proved in the so-called {\it Saha regime}, where large $\beta$ and small $\rho$ are coupled with each other via the equation $\rho=\e^{-\beta\nu}$ for some parameter $\nu\in(0,\infty)$, and the limiting rate function $g_\nu$ turned out to be piecewise linear with at least one kink, but also possibly more. Each kink represents a phase transition, and the minimiser in the kinks is not unique. This is in strong contrast with the approximate rate function $f^\ideal(\beta,\rho,\cdot)$ studied in this paper, which possesses only one minimiser and only (at most) one phase transition. As we make explicit in the next paragraph, $f^\ideal(\beta,\rho,\cdot)$ can itself be approximated by $g_\nu$, in particular by neglecting an entropy term. It is the smoothing effect of this term that gets lost in that approximation, and possibly a lot of new kinks appear in this way. We know that these additional kinks correspond to cross-overs inside the gas phase, but not to sharp phase transitions (see \cite{J} for a discussion of this). Hence, the full ideal mixture captures the behaviour of the physical system much better than the function studied in \cite{jkm}.

We note that both $f^\ideal(\beta,\rho,\cdot)$ and $g_\nu$ considered in~\cite{CKMS10,jkm} appear as exact large deviations rate functions for simple random partitions models. We consider vectors $(N_k)_{1\leq k \leq N} \in \N_0^N$ with $\sum_{k=1}^N k N_k = N$ as integer partitions, and look at the (not normalised) measures on partitions given by
\begin{equation}\label{twomodels}
\begin{aligned}
\mu^\ideal_{\beta,N,\Lambda}( \{(N_1,\ldots,N_N)\}) & := \prod_{k=1}^N \frac{(|\Lambda| Z_k^\cl(\beta))^{N_k}}{N_k!}
		 = \binom{M}{N_1,\ldots,N_N} 
				\times \frac{|\Lambda|^M}{M!} \prod_{k=1}^N ( Z_k^\cl(\beta))^{N_k},\\ 
\mu^{\sf{CKMS}}_{\beta,N,\Lambda} (\{N_1,\ldots,N_N\}) &:= 
		\frac{|\Lambda|^M}{M!} \prod_{k=1}^N (  \e^{-\beta E_k}) ^{N_k}, 
\end{aligned}
\end{equation}
where $M:= \sum_{k=1}^N N_k$. In the thermodynamic limit $N,|\Lambda|\to \infty$ such that $N/|\Lambda|\to \rho$, under $\mu^\ideal_{\beta,N,\Lambda}$, the vector $(N_k/|\Lambda|)_{k\in \N}$ satisfies a large deviations principle with speed $\beta |\Lambda|$ and rate function $f^\ideal(\beta,\rho,\cdot)$. On the  other hand, under $\mu^{\sf{CKMS}}_{\beta,N,\Lambda}$, the vector $( k N_k/N)_{k\in\N}$ satisfies a large deviations principle with speed $\beta N$ and rate function $\vect{q}\mapsto g_\nu(\vect{q})- \frac{1}{\beta} \sum_{k\in\N} \frac{q_k}{k}$, which differs from the approximate functional $g_\nu$ from~\cite{CKMS10,jkm} only by a vanishing term. Hence, from \eqref{twomodels} we see that $g_\nu(\cdot)$ arises from $f^ \ideal(\beta,\rho,\cdot)$ by two simplifications: $Z_k^\cl(\beta) \approx \exp(-\beta E_k)$ and the omission of the multinomial coefficient, that is, the cluster free energy is approximated by the ground state energy, and the mixing entropy is neglected.

The second main difference with~\cite{jkm} is that our error bounds are much better. 
In~\cite{jkm},  the free energy per particle $f(\beta,\rho)/\rho$ is approximated up to errors 
of the order $(\log \beta)/\beta$. In contrast, when  
$\rho = \e^{-\beta \nu}$ for fixed $\nu>0$, as $\beta \to \infty$, the error in \eqref{eq:expoa} vanishes exponentially fast in $\beta$. Moreover, \eqref{eq:expoa} may be written as 
\begin{equation*}
	f(\beta,\rho) = \rho\mu^\ideal(\beta,\rho) -\frac{1}{\beta} m^\ideal(\beta,\rho) 
		\Bigl( 1+ O(\rho^{1/(d+1)}) \Bigr). 
\end{equation*}
This is interesting because for $\rho>\rho_\sat^\ideal(\beta)$, we have $m^\ideal(\beta,\rho) \leq \rho_\sat^\ideal(\beta)$. Thus  the free energy equals the ideal free energy plus an error which is small compared to the smallest of the two terms in~\eqref{eq:fideal}, $\frac 1\beta m^\ideal$. 

For completeness and for the reader's convenience, in Section~\ref{sect-Saha}, we provide approximations of the idealised mixture model in terms of $g_\nu$ in the Saha regime with exponentially small errors.

\section{Proof of Theorem~\ref{thm:bounds}}\label{sect-proofthmbounds}

\noindent In this section, we prove Theorem~\ref{thm:bounds}. We will use a convexity argument and split an arbitrary sequence $\vect{\rho}$ into its components on the first $K$ entries and the ones on the remainder. Hence, we consider these two parts separately.

\subsection{Case 1: all clusters have size $\boldsymbol{\leq K}$} 

Here we consider $\vect{\rho}$ satisfying $\rho_k=0$ for $k>K$ with some $K\in\N$. Recall from Assumption (V) that $b$ is the interaction range and $R\in(b,\infty)$ determines the notion of connectedness.

\begin{lemma} \label{lem:fkcomp}
	Let $\beta>0$, $k \in \N$. Set $C:= \sup_{|x|\leq 2/3} |x^{-1}\log (1-x)|$.
	Then 
	\begin{itemize}
		\item[(i)] For all $A>0$, $ f_k^\cl(\beta) \leq f_k^{\cl,A}(\beta)$. 
		\item[(ii)] For all  $A>3k R$, 
		\begin{equation*}
			 f_k^{\cl,A}(\beta) \leq  f_k^{\cl}(\beta) + \frac{C d R}{\beta A}. 
		\end{equation*}
	\end{itemize}
\end{lemma}

\begin{proof}
	(i) For all $\beta,k,A$, we have  $Z_k^{\cl,A}(\beta) \leq Z_k^\cl(\beta)$, which implies that
	$f_k^{\cl,A}(\beta) \geq f_k^\cl(\beta)$.

(i) Let $k \in \N$ and $a >2kR$. Let $x_1 \in [0,a]^d$ have distance $\geq k R$ to the 
	boundary of the cube. Then, writing $\vect{x}=(x_1,\dots,x_k)$,
	$$
\begin{aligned}
		\frac{1}{k!} \int_{[0,a^d]^{k-1}} \e^{- \beta U_k(\vect{x})}\1\{\vect{x}\ \text{connected}\}\, \dd x_2 \cdots \dd x_k 
		&= \frac{1}{k!} \int_{(\R^d)^{k-1}} \e^{- \beta U_k(\vect{x})}\1\{\vect{x}
		\ \text{connected}\}\, \dd x_2 \cdots \dd x_k \\
&= Z_k^\cl(\beta).
	\end{aligned}
$$
Thus, integrating over $x_1$ and recalling the definition of $Z_k^{\cl,A}(\beta)$ in \eqref{Zcladef},
	\begin{equation*}
		Z_k^{\cl,A}(\beta) \geq \frac{ (A- 2 k R)^d}{A^d} Z_k^\cl(\beta). 
	\end{equation*}
	Therefore, when $A>3 kR$, we take $-\frac 1{\beta k}\log$ and deduce
	$$f_k^{\cl,A}(\beta) \leq f_k^{\cl}(\beta) - \frac{1}{\beta k} \log\big[(1- \smfrac{2k R}A)^d\big]  \leq f_k^{\cl}(\beta) + \frac{C d R}{\beta A}. $$
\end{proof}

\begin{lemma} \label{lem:ubfinite}
	Fix $\beta,\rho>0$. Let $K \in \N$ and $A>0$ such that 
	 $(A+R)^d <\rho^{-1}$. Then for all 
	 $\vect{\rho}$ such that $\sum_{k=1}^K k \rho_k = \rho$, 
	\begin{equation} \label{eq:ubfinite}
	\begin{aligned}
		f(\beta,\rho,\vect{\rho}) & \leq \sum_{k=1}^K k \rho_k f_k^{\cl,A}(\beta) 
			+ \frac{1}{\beta} \sum_{k=1}^K \rho_k (\log \rho_k - 1) \\
			& \quad + \frac{1}{\beta} \left( \sum_{k=1}^K \rho_k \right) 
			\left( - \log \Bigl(1- (A+R)^d \sum_{k=1}^K \rho_k\Bigr) 
		+ \log \Bigl( 1+ \smfrac{R}{A} \Bigr)^d \right). \qedhere
	\end{aligned}
	\end{equation}
\end{lemma}

\begin{proof}
	Let $\Lambda = [0,L]^d$. For $2 \leq k \leq K$, set $N_k:= \lfloor |\Lambda|\rho_k \rfloor$. 
	Set $N_1:= N - \sum_{k=2}^K k N_k$, and $M:= \sum_{k=1}^K N_k$. 	Divide $\Lambda$ into cubes of side-length $A$, at mutual distance $R$. We call these cubes ``cells''.
	The number $D$ of cubes that can be placed in this way is $D \geq \lfloor L / (A+R)\rfloor^d$. 
	Let 
	\begin{equation*}
	Z_{\Lambda}(\beta,N,N_1,\ldots,N_K):= \frac{1}{N!} \int_{\Lambda^N} \e^{-\beta U_N(\vect{x})} \1\big\{ \forall k\in \{1,\ldots,K\}:\, N_k(\vect{x}) = N_k\big\}\, \dd \vect{x}. 
	\end{equation*}
	We can lower bound this constrained partition function by integrating only over 
	configurations such that: (a) there is at most one cluster per cell, and (b) 
	each cluster is contained in one of the cells. The number of partitions of the particle label set $\{1,\ldots,N\}$ 
	into $N_1$ sets of size $1$, $N_2$ sets of size $2$, etc., is equal to 
	$N! /( \prod_{k=1}^K k! N_k!)$. For a given set partition, the number of ways 
	to assign distinct cells to the sets is $D (D-1) \cdots (D-M+1)$. Thus we find 
	\begin{equation*}
		Z_{\Lambda}(\beta,N,N_1,\ldots,N_K) \geq D (D-1) \cdots (D-M+1) \prod_{k=1}^K 
			\frac{\bigl( A^d Z_k^{\cl,A}(\beta) \bigr)^{N_k}}{N_k!}. 
	\end{equation*}
	We observe that 
	\begin{equation*}
		\begin{aligned} 
			D (D-1) \cdots (D-M+1) (A^d)^M 
			& \geq \Bigl( \bigl( \smfrac{L}{A+R} - 1\bigr)^d - M \Bigr)^M (A^d)^M \\
			& = |\Lambda|^M \Bigl( \frac{A}{A+R} \Bigr)^{dM} 
				\Bigl( \bigl( 1- \smfrac{A+R}{L} \bigr)^d - \frac{M (A+R)^d}{|\Lambda|} \Bigr) ^M. 
		\end{aligned}
	\end{equation*}
	It follows that in the limit $N,L \to \infty$, $N/L^d \to \rho$, 
	$- \liminf \frac{1}{\beta |\Lambda|} \log Z_\Lambda(\beta,N,N_1,\ldots,N_K)$
	is not larger than the right-hand side of \eqref{eq:ubfinite}. 
	The proof of the lemma is then concluded as in the proof of~\cite[Proposition 3.2]{jkm}. 
\end{proof}

\begin{prop}\label{prop:case1}
	Fix $\beta >0$, 
	$\rho\leq (2^{d+1}/3)^{1+1/d}$, and $K\in \N$ with $K< \rho^{-1/(d+1)}$. 
	Then, for suitable $C'>0$ and all $\vect{\rho}$ with $\sum_{k=1}^K k \rho_k = \rho$ and $\rho_k = 0$ for 
	$k\geq K+1$, 
	\begin{equation*}
		f(\beta,\rho,\vect{\rho}) \leq f^\ideal(\beta,\rho,\vect{\rho}) + 
   			C' \frac{\rho}{\beta} \rho^{1/(d+1)}. 
	\end{equation*}
\end{prop}

\begin{proof}
	Let $C>0$ be as in Lemma~\ref{lem:fkcomp} and set 
	$m:= \sum_{k=1}^K \rho_k$. Then by Lemmas~\ref{lem:fkcomp}
	and~\ref{lem:ubfinite}, 
	\begin{equation*}
		f(\beta,\rho,\vect{\rho}) - f^\ideal(\beta,\rho,\vect{\rho}) 
		\leq \frac{C d R}{\beta A}\rho + \frac{C m}{\beta} \Bigl( (A+R)^d m + \frac{d R}{A} \Bigr),  
	\end{equation*}
	provided $A> 3 KR$ and $(A+R)^d  m \leq 2/3$. 
	Set $A = \rho^{-1/(d+1)}/(3R)$. Then $A>3KR$ because by assumption $K<\rho^{-1/(d+1)}$, and $(A+R)^d m \leq 2^d A^d \rho = 2^d \rho^{d/(d+1)} \leq 2/3$ 
	because we have assumed $\rho\leq (2^{d+1}/3)^{1+1/d}$. Thus we have 
	\begin{equation*}
	f(\beta,\rho,\vect{\rho}) - f^\ideal(\beta,\rho,\vect{\rho}) 
		\leq C' \frac{\rho}{\beta} \rho^{1/(d+1)},\qquad\mbox{where } C':= C \Bigl( \frac{2^d}{(3R)^d}
			+ 6 d R^2\Bigr).   \qedhere
	\end{equation*} 
\end{proof}

\subsection{Case 2: all clusters have size $\boldsymbol{\geq K+1}$} \label{sec:case2}

\begin{prop}\label{prop:case2} 
	Suppose that $v$ satisfies Hypotheses~\ref{ass:hoelder} and~\ref{ass:maxdist}. 
	Then there are $\overline{\rho},\overline{\beta},C''>0$ such that 
	for all $\beta\in[\overline{\beta},\infty)$, $\rho\in(0,\overline{\rho})$, for all $K\in \N_0$ 
	and all $\vect{\rho}$ such that $\sum_{k=K+1}^\infty k \rho_k \leq\rho$
	and $\rho_k=0$ for $k \leq K$, 
	\begin{equation} \label{eq:case2}
		f(\beta,\rho,\vect{\rho}) \leq f^\ideal (\beta,\rho,\vect{\rho})
			+ \frac{C''\rho }{\beta} \log \beta - \frac{m}{\beta}\log\frac{m}{\rho\,\e}. 
	\end{equation}
	If in addition Hypothesis~\ref{collapse} holds, we can replace $\rho \beta^{-1}\log \beta$ 
	with $(\sum_{k=K+1}^\infty \rho_k) \beta^{-1}\log \beta$.
\end{prop}

\begin{proof} 
	By~\cite[Theorem 1.8]{jkm}, there are $\overline \rho, \overline \beta, C''>0$ such that 
	for all $\beta \geq \overline\beta$ and $\rho \leq \overline \rho$ , 
	\begin{equation} \label{eq:thm18}
		f(\beta,\rho,\vect{\rho}) \leq \sum_{k=K+1}^\infty \rho_k \Bigl(  E_k  
			+ \frac{\log \rho}{\beta } \Bigr)
			+ \Big(\rho-  \sum_{k=K+1}^\infty k \rho_k \Big) e_\infty 
			+ \frac{C'' \rho}{\beta} \log \beta. 
	\end{equation}
	Because of~\cite[Lemma 4.3]{jkm}, we can further increase $C''$ and $\overline \beta$ so that 
	for all $k\in \N$ and $\beta \geq \overline \beta$, 
	\begin{equation*}
		\frac{E_k}{k}  \leq  f_k^\cl(\beta) + \frac{C''}{\beta}\log \beta\qquad \mbox{and}\qquad
		e_\infty \leq f_\infty^\cl(\beta) + \frac{C''}{\beta} \log \beta. 
	\end{equation*}
	Moreover,  setting $m:= \sum_{k=K+1}^\infty \rho_k$,
	\begin{equation*}
		\sum_{k=K+1}^\infty \rho_k \frac{\log \rho}{\beta} 
		 - \frac{1}{\beta} \sum_{k= K+1}^\infty \rho_k(\log \rho_k - 1) 
		\leq 
		\frac{m}{\beta} \Bigl( 2 + \log \frac{\rho}{m} + \log \frac{\sum_{k=K+1}^\infty k\rho_k}{m}\Bigr). 
	\end{equation*}
	Here we have used that 
	\begin{align*}
		- \sum_{k=K+1}^\infty \rho_k \log \frac{\rho_k}{m} \leq m \Bigl( 1+ \log 
		\frac{\sum_{k=K+1}^\infty k\rho_k}{m}\Bigr),
	\end{align*}
	see~\cite[Lemma 4.1]{jkm}. The inequality~\eqref{eq:case2} follows.  

	If in addition Hypothesis~\ref{collapse} holds we remark, first, that 
	for all $\rho \leq 1/c$ and all sufficiently large $\beta$, with $c$ as in Hypothesis~\ref{collapse}, 
	$f(\beta,\rho,\vect{0}) \leq f_\infty^\cl(\beta) = f^\ideal(\beta,\rho,\vect{0})$. Because of the convexity of $f(\beta,\rho,\cdot)$, 
	\begin{equation*}
		f(\beta,\rho,\vect{\rho}) \leq \frac{\sum_{k=1}^\infty k \rho_k}{\rho}  f(\beta,\rho, \vect{\tilde \rho}) + \frac{\rho- \sum_{k=K+1}^\infty k \rho_k}{\rho}\, f(\beta,\rho,\vect{0}),\quad \vect{\tilde\rho} := \Big( \frac{\sum_{k=1}^\infty k \rho_k}{\rho} \Big)^{-1} 
		\vect{\rho}.  
	\end{equation*}
	We deduce that in \eqref{eq:thm18} we can replace $e_\infty$ with $f_\infty^\cl(\beta)$, and the claim follows. 
\end{proof}

\subsection{General case} 

\begin{proof}[Proof of Theorem~\ref{thm:bounds}]
	We already know that 
	$f(\beta,\rho,\vect{\rho}) \geq f^\ideal(\beta,\rho,\vect{\rho})$ \cite[Lemma 3.1]{jkm}, 
	so we need only  prove the upper bound for $f(\beta,\rho,\vect{\rho})$. 	Let $K\in \N$ and $\vect{\rho}$ such that $\sum_{k=1}^\infty k \rho_k \leq \rho$. 
	The  cases  $\sum_{k=1}^k\rho_k = \rho$ and $=0$ were treated in Propositions~\ref{prop:case1} and~\ref{prop:case2}, 
	thus we may assume $0<\sum_{k=1}^K k \rho_k < \rho$. 
	Set $\rho_{\leq K} := \sum_{k=1}^K k \rho_k$. 
	Let $\overline{\rho}$ be as in Proposition~\ref{prop:case2}  and suppose that 
	$\rho \leq \overline \rho$. Set  
	\begin{equation*}
		\underline{\rho} := \frac{\rho_{\leq K}}{1 - \rho_{\leq K}/\overline{\rho}}
	\end{equation*}
	and  
	\begin{equation*}
		\rho_k^\text{small}  := \begin{cases}
				\rho_k / (1 - \rho_{\leq K}/\overline{\rho}) &\mbox{if } 1 \leq k \leq K,\\
				0, &\mbox{if } k \geq K+1, 
		          \end{cases} \qquad \mbox{and}\qquad
		\rho_k^\text{large}  := \begin{cases}
				0 &\mbox{if } 1 \leq k \leq K,\\
				\rho_k / [\rho_{\leq K}/\overline{\rho}], &\mbox{if } k \geq K+1. 
		          \end{cases}
	\end{equation*}
It was shown in \cite[Section 2.5]{jkm} that the map $(\rho,\vect{\rho})\mapsto f(\beta,\rho,\vect{\rho})$ is a supremum of convex functions and hence is itself convex. Thus we can write 
	\begin{equation} \label{eq:ft}
		f(\beta,\rho,\vect{\rho}) \leq \Bigl(1-\frac{\rho_{\leq K}}{\overline{\rho}}\Bigr) f\bigl(\beta,\underline{\rho}, \vect{\rho}^\text{small} \bigr) + 
			\frac{\rho_{\leq K}}{\overline{\rho}}\, f\bigl( \beta,\overline{\rho}, \vect{\rho}^\text{large} \bigr). 
	\end{equation}
	Propositions~\ref{prop:case1} and~\ref{prop:case2} yield 
	\begin{equation*} \label{eq:ubound} 
	\begin{aligned}
		f(\beta,\rho,\vect{\rho}) - f^\ideal(\beta,\rho,\vect{\rho})  
		& \leq  C' \frac{\rho_{\leq K}}{\beta} 
		\Bigl(\sum_{k=1}^K \rho_k^\text{small} \Bigr)^{1/(d+1)}
			- \frac{m_{\leq K}}{\beta} \log \bigl( 1- \frac{\rho_{\leq K}}{\overline{\rho}} \bigr)\\
			& \qquad + \frac{C''}{\beta} (\rho- \rho_{\leq K}) \log \beta  - 2 \frac{m_{>K}}{\beta}
		\log \frac{m_{>K}}{\overline{\rho}\,\e}, 
	\end{aligned}
	\end{equation*}
	provided $\beta \geq \overline{\beta}$, $\underline{\rho} \leq (2^{d+1}/3)^{1+1/d}$, 
	and $K< \underline{\rho}^{-1/(d+1)}$.  The conditions on $\underline{\rho}$ and $K$ are certainly satisfied if we assume 
	that $\rho \leq \frac{2}{3}\min(1,\overline{\rho})$ and $K \leq (\rho /3)^{-1/(d+1)}$. 
	If this is the case, we can further bound the right-hand side of \eqref{eq:ubound}
	as 
	\begin{equation*}
		\frac{C' 3^{1/(d+1)} + C}{\beta} \bigl( \rho_{\leq K} \bigr)^{(d+2)/(d+1)} 
			+ \frac{C''}{\beta} (\rho- \rho_{\leq K}) \log \beta  - 2 \frac{m_{>K}}{\beta}
		\log \frac{m_{>K}}{\overline{\rho}\,\e}. 
	\end{equation*}
	Since $m_{>K} \leq \rho \leq 2/3$ , we can upper bound 
	$|1+ \log \overline{\rho}|\leq - D(\overline{\rho}) \log m_{>K}$ for some suitable 
	constant $C(\overline{\rho}) >0$. 
	We obtain the bound of Theorem~\ref{thm:bounds}. The improved bound under Hypothesis~\ref{collapse} is 
	deduced from the corresponding statement in Proposition~\ref{prop:case2}. 
\end{proof}

\section{Proof of Theorems~\ref{thm:exponential} and~\ref{thm:ideal}}

We prove Theorems~\ref{thm:exponential} and~\ref{thm:ideal} in Sections~\ref{sect-proofthmexpo} and \ref{sect-proofthmideal} below, respectively, after providing some preparations in Section~\ref{sect-tailesti}.

\subsection{Tail estimates}\label{sect-tailesti}

The main idea for the proof of Theorems~\ref{thm:exponential} and~\ref{thm:ideal} 
is to apply Theorem~\ref{thm:bounds} to the minimiser $\vect{\rho}^\ideal(\beta,\rho)$ of the ideal rate function. Therefore we will need estimates on the tails of the ideal minimiser, see  Lemmas~\ref{lem:idealtail2} and~\ref{lem:idealtail} below. 
First we recall a statement about the low-temperature behaviour of the cluster free energy. Recall the ground-state energy $E_k$ defined in \eqref{Ekdef}.

\begin{lemma}[Cluster free energy at low temperature] \label{lem:fk-lowtemp}
	There are $\beta_0>0$ and $C>0$ such that 
	for $\beta \geq \beta_0$, 
	\begin{equation*}
		\forall k \in \N:\ f_k^\cl(\beta) \geq \frac{E_k}{k} - \frac C \beta,\qquad\mbox{and}\qquad 
		f_\infty^\cl(\beta) \geq e_\infty - \frac C\beta.
	\end{equation*}
	Moreover, for each fixed $k \in \N$, 
	$\lim_{\beta \to \infty} f_k^\cl(\beta) = E_k/k$. 
	If in addition the pair potential satisfies Hypothesis~\ref{ass:hoelder}, 
	then for $\beta \geq \beta_0$, 
	\begin{equation*}
		\forall k \in \N:\ f_k^\cl(\beta) \leq \frac{E_k}{k} + \frac C \beta \log \beta, \qquad\mbox{and}\qquad f_\infty^\cl(\beta) \leq e_\infty + \frac C\beta\log \beta.
	\end{equation*}
	and $\lim_{\beta \to \infty} f_\infty^\cl(\beta) = e_\infty$. 
\end{lemma}

\begin{proof}
The limit statement for $f_k^\cl(\beta)$ as $\beta \to \infty$ follows by a standard argument for exponential integrals. The lower bounds for $f_k^\cl(\beta)$ and $f_\infty^\cl(\beta)$ have been proven in \cite[Lemma 4.3]{jkm}. The upper bounds follow from~\cite[Lemma 4.5]{jkm} and the observation that for all $\beta,a>0$ and $k\in \N$, $f_k^\cl(\beta) \leq f_k^{\cl,a}(\beta)$, see Lemma~\ref{lem:fkcomp}.
\end{proof}

In the following, we omit the arguments $(\beta,\rho)$ of the objects $\vect{\rho}^\ideal$ and $m^\ideal$ for brevity. Recall that $\vect{\rho}^\ideal$ is given in \eqref{eq:rhok} and that $m^\ideal=\sum_{k\in\N}\rho^\ideal_k$ and $m_{>K}^\ideal=\sum_{k=K+1}^\infty \rho^\ideal_k$.

\begin{lemma} \label{lem:idealtail2}
Suppose that Hypothesis~\ref{secondorder} is true, and assume that $d\geq 2$. Then there are $C,c,\overline\beta>0$, $K_0 \in \N$ such that for all $K \geq K_0$, $\beta \in[\overline\beta,\infty)$, and $\rho >0$, 
  \begin{equation}\label{rhoidealesti}
		\frac{\sum_{k\in \N} k \rho_k^\ideal}{\sum_{k\in \N} \rho_k^\ideal} \leq C,\qquad\mbox{and}\qquad
      \frac{m_{>K}^\ideal}{m^\ideal} 
      \leq \frac{\sum_{k= K+1}^\infty k \rho_k^\ideal}{m^\ideal} 
    \leq \e^{ - \beta c K^{1-1/d}}.
  \end{equation}
\end{lemma}

\begin{proof}
According to Hypothesis~\ref{secondorder}, we can choose $c>0$, $K_0\in \N$, and $\overline\beta$ such that for $\beta \geq \overline\beta$ and $K \geq K_0$, 
	\begin{equation*}
     \sum_{k=K+1}^{\infty} k \exp( \beta k [f_\infty^\cl(\beta) - f_k^\cl(\beta)]) \leq \e^{ - \beta c K^{1-1/d}}.
  \end{equation*}
  Let $k(\nu^*)$ be such that $E_{k(\nu^*)} - k(\nu^*) e_\infty = \nu^*$. Recall from \eqref{eq:rhok} that $\rho_k^\ideal(\beta,\rho) \leq \exp(\beta k[f_\infty^\cl(\beta) - f_k^\cl(\beta)])$. 
  Then, for $K \geq \max\{K_0,k(\nu^*)\}$, we have, also using Lemma~\ref{lem:fk-lowtemp},
$$   
\frac{\sum_{k= K+1}^\infty k \rho_k^\ideal}{\sum_{k=1}^\infty \rho_k^\ideal} 
     \leq \frac{\e^{- \beta c K^{1-1/d}}}{ \rho_{k(\nu^*)}^\ideal} 
     = \frac{ \e^{- \beta c K^{1-1/d}}}{\e^{- \beta (\nu^* + O( \beta ^{-1} \log \beta))}} \leq \e^{ - \beta (c K^{1-1/d} - 2\nu^*)},
$$
where the last inequality holds for sufficiently large $\beta$. For $K$ sufficiently large, this last expression is  smaller than $\e^{- \beta c K^{1-1/d}/2}$. Writing $c/2$ instead of $c$, this shows the second assertion in \eqref{rhoidealesti}. Furthermore, for these $K$, 
     \begin{equation*}
	\frac{\sum_{k\in \N} k\rho_k^\ideal}{m^\ideal} \leq \frac{\sum_{k=1}^K k \rho_k^\ideal}
      {\sum_{k=1}^K \rho_k^\ideal} + \e^{ - \beta c K^{1-1/d}/2} \leq K +1.
    \end{equation*}
Using this for $K=K_0$, this also shows the first assertion in \eqref{rhoidealesti}.
\end{proof}

If we do not assume that Hypothesis~\ref{secondorder} is true, we have an analogue 
of Lemma~\ref{lem:idealtail2} for densities much smaller than $\exp(-\beta \nu^*)$, valid for all $d\in\N$. 

\begin{lemma} \label{lem:idealtail}
	Let $v$ satisfy Hypothesis~\ref{ass:hoelder}.
	Fix $\eps >0$. Then there are $\beta_\eps,K_\eps, \delta_\eps>0$ such that 
	for all $\beta \in[\beta_\eps,\infty)$, $\rho \leq \exp(- \beta (\nu^*+\eps))$ and 
	 $K \geq K_\eps$, 
	\begin{equation*}
		\frac{\rho}{m^\ideal} \leq C_\eps,\qquad\mbox{and}\qquad  
		 \frac{m_{>K}^\ideal}{m^\ideal} \leq \frac{\sum_{k=K+1}^\infty k \rho_k^\ideal}{m^\ideal} \leq \e^{- \beta K \delta_\eps}.
	\end{equation*}
\end{lemma}

For the proof, we give first a lower bound for the saturation density of the ideal mixture, which is of interest in itself. Under additional hypotheses, a stronger statement holds, see Prop.~\ref{prop:satsaha}. 

\begin{lemma}
	\label{lem:rsat}
	$$ \liminf_{\beta \to \infty} \frac 1\beta \log \rho_\sat^\ideal(\beta) \geq - \nu^* . $$ 
\end{lemma}

\begin{proof}
	Let $k \in \N$. We pick, as a lower bound, only the $k$-th summand in \eqref{rhosatdef}. Then, as $\beta \to \infty$, according to Lemma~\ref{lem:fk-lowtemp},
	\begin{align*}
		\frac 1\beta \log \rho_\sat^\ideal(\beta)  \geq \frac 1\beta \log k + k (f_\infty^\cl(\beta) - f_k^\cl(\beta)) = - (E_k - k e_\infty) + o(1).
	\end{align*}
	Letting $\beta \to \infty$ and taking the supremum over $k$ of the right-hand side yields the desired result. 
\end{proof}

\begin{proof}[Proof of Lemma~\ref{lem:idealtail}]
	Fix $\eps>0$. By Lemma~\ref{lem:rsat}, there is a $\beta_0>0$ such that 	for all $\beta \geq \beta_0$, $\e^{-\beta (\nu^*+\eps)}< \rho_\sat^\ideal(\beta)$. 
	Thus for $\beta \geq \beta_0$ and $\rho\leq \e^{ - \beta(\nu^*+\eps)}$, 	 we have
	$\rho < \rho_\sat^\ideal(\beta,\rho)$, and $\mu^\ideal(\beta,\rho)$ 
	solves \eqref{eq:musol}.

	Let $\mu_\eps:= \inf_{k\in \N} (E_k - \nu^*-\eps)/k$. Then $\mu_\eps<e_\infty$. 
	Indeed, by definition of $\nu^*$, there is a $k\in\N$ such that 
	$E_k - ke_\infty \leq \nu^*+\eps/2$, and for this $k$, 
	$\mu_\eps \leq (E_k - \nu^*-\eps) /k \leq e_\infty - \eps/(2k) < e_\infty$. 
	Choosing $k_0 \in \N$ such that $\mu_\eps = (E_{k_0} - \nu^* - \eps) /{k_0}$, 
	we find 
	\begin{equation*}
		\mu^\ideal(\beta,\rho) \leq f_{k_0}^\cl(\beta) + \frac{\beta^{-1}\log \rho}{k_0} \leq \frac{E_{k_0} - (\nu^*+\eps)}{k_0} + O(\beta^{-1} \log \beta) = \mu_\eps +O(\beta^{-1}\log \beta). 
	\end{equation*}
	Since $\mu_\eps<e_\infty$ and $[E_k - (\nu^*+\eps)]/k \to e_\infty$ as $k\to \infty$, 
	there  are $k_\eps \in \N$, $\Delta_\eps >0$ such that 
	\begin{equation*}
		k \geq k_\eps\qquad  \Longrightarrow\qquad  
			\frac{E_k - (\nu^*+\eps)}{k} \geq \mu(\nu^* + \eps) + \Delta_\eps.
	\end{equation*}
	 It follows that 
	 for $k \geq k_\eps$, 
	\begin{equation*}
		\frac{\rho_k^\ideal(\beta,\rho)}{\rho} 
			\leq \exp( - \beta k (\Delta_\eps +O(\beta^{-1} \log \beta))).
	\end{equation*}
	Choose $\beta_\eps\geq \beta_0$ large enough so that the $O(\beta^{-1} \log \beta)$ term is $\leq \Delta_\eps / 2$, then we find for $\beta \geq \beta_\eps$ 
	and $k \geq k_\eps$, 
	\begin{equation*}
		\frac{\rho_k^\ideal(\beta,\rho)}{\rho} \leq \exp( - \beta k \Delta_\eps / 2).
	\end{equation*}
	Noting that for $z<1$, as 
	$K \to \infty$, 
	$\sum_{k \geq K} k z^k = O(z^K),$
	we deduce that for sufficiently large $K$, 
	\begin{equation*}
	\frac{\sum_{k=K+1}^\infty
		 k \rho_k^\ideal(\beta,\rho)}{\rho } \leq \exp(- \beta K \delta_\eps).
	\end{equation*}
	Now fix 
	$K_1 \geq K_\eps$ large enough so that $\exp( - \beta_\eps \delta_\eps K_1) \leq 1/2$. 
	Then 
	\begin{equation*}
		\rho = \sum_{k=1}^{K_1} k \rho_k^\ideal + \sum_{k=K_1+1}^\infty  k \rho_k^\ideal(\beta,\rho) 
			\leq K_1 m^\ideal + \rho /2
	\end{equation*}		
	whence $ \rho \leq 2 K_1 m^\ideal$ and, for sufficiently large $K$, 
 	\begin{equation*}
	\frac{\sum_{k=K+1}^\infty
		 k \rho_k^\ideal}{m^\ideal} \leq 2 K_1 \exp(- \beta K \delta_\eps). \qedhere
	\end{equation*}
\end{proof}

\subsection{Proof of Theorem~\ref{thm:exponential}}\label{sect-proofthmexpo}

\subsubsection{Free energy $f(\beta,\rho)$} 

By Theorem~\ref{thm:bounds}, for all $\beta$ and $\rho$, 
$$f(\beta,\rho) = \inf_{\vect{\rho}} f(\beta,\rho,\vect{\rho}) \geq \inf_{\vect{\rho}} f^\ideal(\beta,\rho,\vect{\rho}) = f^\ideal(\beta,\rho). $$
Thus we need only prove the upper bound to $f(\beta,\rho)$. To this aim we note that 
\begin{equation*}
	f(\beta,\rho)  \leq f(\beta,\rho, \vect{\rho}^\ideal(\beta,\rho)) 
			= f^\ideal(\beta,\rho) + \Bigl( f(\beta,\rho,\vect{\rho}^\ideal(\beta,\rho)) 
			- f^\ideal(\beta,\rho,\vect{\rho}^\ideal(\beta,\rho)) \Bigr).
\end{equation*}
where we recall that $\vect{\rho}^\ideal(\beta,\rho)$ is the unique minimiser of $f^\ideal(\beta,\rho,\cdot)$. To lighten notation, we will drop the $(\beta,\rho)$-dependence in the notation 
and write $\vect{\rho}^\ideal$, $m^\ideal$,  instead of $\vect{\rho}^\ideal(\beta,\rho)$, 
$m^\ideal(\beta,\rho)$, etc. Theorem~\ref{thm:bounds} yields 
\begin{equation*}
	f(\beta,\rho,\vect{\rho}^\ideal) 
			- f^\ideal(\beta,\rho,\vect{\rho}^\ideal) 
	\leq \frac{C}{\beta} \Bigl( (\rho_{\leq K}^\ideal)^{(d+2)/(d+1)}
	+ \sum_{k=K+1}^\infty k \rho_k^\ideal\log \beta    - m_{>K}^\ideal \log m_{>K}^\ideal \Bigr),
\end{equation*}
provided $\beta \geq \overline{\beta}$, $\rho <\overline{\rho}$, and $K<(\rho/3)^{-1/(d+1)}$. 
By Lemma~\ref{lem:idealtail2}, 
if we assume $\beta \geq \beta_0$ and $K\geq K_0$, 
the term in the big parenthesis 
is bounded above by a constant times 
\begin{equation*}
	m^\ideal \Bigl( \rho^{1/(d+1)} + \e^{ - \beta c K^{1-1/d}} \log \beta + \beta c 
	K^{1-1/d} \e^{-\beta c K^{1-1/d}} \Bigr). 
\end{equation*}	
Choosing $K$ as a constant times $\rho^{-1/(d+1)}$, we see that the second and third 
summands are bounded by the first. This gives the desired bound on the free energy. 
\qed

\subsubsection{Minimisers} 
Every minimiser $\vect{\rho}$ of $f(\beta,\rho, \cdot)$ satisfies $f(\beta,\rho) = f(\beta,\rho, \vect{\rho}) \geq f^\ideal(\beta,\rho,\vect{\rho})$. Therefore, according to part (1),
\begin{equation*}
0 \leq f^\ideal(\beta,\rho,\vect{\rho}) - f^\ideal(\beta,\rho) \leq f(\beta,\rho) - f^\ideal(\beta,\rho)\leq \frac C \beta m^\ideal \rho^{1/(d+1)}.
\end{equation*}
Now by an explicit computation, 
\begin{equation*}
	     \beta f^\ideal(\beta,\rho, \vect{\rho}) - \beta f^\ideal(\beta,\rho) \\
% 		& \qquad = \Bigl( \rho_\infty \beta f_\infty^\cl(\beta) + \sum_{k\in \N} \rho_k (\beta k f_k^\cl(\beta) 
% 			+ \log \rho_k - 1) \Bigr) - \Bigl( \rho \beta \mu^\ideal(\beta,\rho) - \sum_{k\in \N} 
% 		      \rho_k^\ideal(\beta,\rho) \Bigr) \\
% 		& \qquad = \sum_{k\in \N} \Bigl( \rho_k^\ideal(\beta,\rho) - \rho_k + \rho_k \log 
% 		\frac{\rho_k}{\rho_k^\ideal(\beta,\rho)} \Bigr) + \rho_\infty (f_\infty^\cl(\beta) 
% 		  - \mu^\ideal(\beta,\rho))\\
		= H(\vect{\rho};\vect{\rho}^\ideal) + \rho_\infty (f_\infty^\cl(\beta) 
		  - \mu^\ideal),
\end{equation*}
where $\rho_\infty:=\rho-\sum_{k=1}^\infty k \rho_k$. 
Let $p_k:= \rho_k / m$ and $p_k^\ideal:= \rho_k^\ideal / m^\ideal$. Then 
\begin{equation*}
   H(\vect{\rho};\vect{\rho}^\ideal) = m^\ideal g\big(\smfrac  m{m^\ideal}\big) + m H(\vect{p};\vect{p}^\ideal),
\end{equation*}
where $g(x):= 1- x + x \log x$. Note that $g(x) \geq 0$ for all $x>0$. Summarizing the last three displays, we find, 
$$
m^\ideal g\big(\smfrac  m{m^\ideal}\big) + m H(\vect{p};\vect{p}^\ideal)+ \rho_\infty (f_\infty^\cl(\beta) 
		  - \mu^\ideal)\leq C m^\ideal \rho^{1/(d+1)}.
$$
Since each of the three terms on the left-hand side is nonnegative, we obtain that each of them is not larger than the right-hand side, and this implies
\begin{equation*}
	 g\big(\smfrac  m{m^\ideal}\big) \leq C  \rho^{1/(d+1)} \qquad \mbox{and}\qquad
	H(\vect{p};\vect{p}^\ideal)  \leq C \frac{m^\ideal}{m} \rho^{1/(d+1)}.
\end{equation*}
Since $g(x) \to 0$ implies $x \to 1$ and $g(x) \sim (1-x)^2/2$ as $x\to 1$, 
we deduce that for $\rho$ sufficiently small and some suitable constant $C'>0$, 
$|1- m /m^\ideal| \leq C' \rho^{1/(2d+2)}$, and the corresponding bound for the 
relative entropy easily follows.  
\qed

\subsection{Proof of Theorem~\ref{thm:ideal}}\label{sect-proofthmideal} 

The argument is exactly the same as for Theorem~\ref{thm:exponential}. 
In Theorem~\ref{thm:bounds} applied to $\vect{\rho}^\ideal$, we note that 
for $\rho\leq \rho_\sat^\ideal$, we have 
$\sum_1^\infty k \rho_k^\ideal =\rho$; this observation replaces the use of Hypothesis~\ref{collapse}. 
Later we use Lemma~\ref{lem:idealtail}  instead of 
Lemma~\ref{lem:idealtail2}. 

\appendix

\section{The idealised model in the Saha regime}\label{sect-Saha} 

\noindent In this section, we provide explicit bounds for the approximation of the minima and the minimisers of the idealised rate function $f^\ideal$ by the ones of the function $g_\nu$ that we introduced in \eqref{eq:ckms} and analysed in \cite{CKMS10} and \cite{jkm}. We work in the Saha regime, where $\rho=\e^{-\beta\nu}$ for some $\nu\in(0,\infty)$. Recall that the interaction potential $v$ is always supposed to satisfy Assumption (V). Let us first recall the relevant notation. 

The ground state energy $E_k$ was defined in~\eqref{Ekdef} and the quantities 
 $e_\infty= \lim_{k\to \infty} E_k/k$ and $\nu^*= \inf_{k\in \N}(E_k - ke_\infty)$  
were defined after~\eqref{Ekdef}. Set $\mu(\nu) := \inf_{k\in \N} [E_k- \nu]/k$. 
 From \cite[Lemma~1.3]{jkm} we know that the map $\nu\mapsto \mu(\nu)$ is piecewise affine. 
It is constant with value $\mu(\nu) = e_\infty$ for $\nu \in (0,\nu^*]$, and strictly decreasing in 
$[\nu^*,\infty)$. The set $\Ncal\subset[\nu^*,\infty)$ of points at which $\mu$ changes its slope is bounded and either infinite, with the unique accumulation point $\nu^*$, or finite. 
 Furthermore, for $\nu \in(\nu^*,\infty)\setminus \mathcal{N}$, we have $\mu(\nu) = [E_{k_\nu} - \nu]/k_\nu$ for a unique $k_\nu \in \N$, and 
\begin{equation}\label{Deltadef}
\Delta(\nu):= \inf \Bigl \lbrace \frac{E_k - \nu}{k} - \mu(\nu) 
			\ \Big |\ k \in \N,\ k \neq k_\nu \Bigr \rbrace 
\end{equation}
is strictly positive~\cite[Theorem 1.8]{jkm}.

A first quick consistency check concerns the comparison of the critical line $\rho = \exp(-\beta \nu^*)$ from~\cite{CKMS10,jkm} with the saturation density of the ideal mixture; this strengthens Lemma~\ref{lem:rsat}. \begin{prop}[Saturation density] \label{prop:satsaha}
	Suppose that $v$ satisfies Hypotheses~\ref{ass:hoelder}, \ref{ass:maxdist} and~\ref{secondorder} and $d\geq 2$.  
	Then, as $\beta \to \infty$, 
	$\rho_\sat^\ideal(\beta) = \exp( - \beta \nu^* + O(\log \beta))$. 
\end{prop} 

Next, we investigate the low-temperature asymptotics of $f^\ideal(\beta,\rho)$. 
Recall that the free energy is a sum of two terms, see~\eqref{eq:fideal}. We analyse them separately and  shall see that the dominant contribution comes from the term $\rho \mu^\ideal(\beta,\rho)$, which behaves like $\rho \mu(\nu)$. Observe that 
$\rho \mu(\nu)$ is precisely the approximation to the free energy $f(\beta,\rho)$ proven in~\cite{jkm}. 

\begin{prop}[Chemical potential] \label{lem:chempot}
	Suppose that $v$ satisfies Hypotheses~\ref{ass:hoelder} and \ref{ass:maxdist}. 
	Let $\nu >0$ and put $\rho = \exp(-\beta \nu)$. 
	Then,  as $\beta\to\infty$,
	\begin{itemize}
		\item if $\nu\in(\nu^*,\infty)\setminus \mathcal{N}$,
			\begin{equation}  \label{eq:musaha}
				\mu^\ideal(\beta,\rho) 
			 = f_{k_\nu}^\cl(\beta) - \frac{\nu}{k_\nu} - \frac{ \log k_\nu}{\beta} + O(\beta^{-1} \e^{-\beta \Delta(\nu)/2}) 
			=  \mu(\nu) + O\Big(\frac{\log \beta}{\beta}\Big),
			\end{equation} 
		\item if $\nu < \nu^*$ and $v$ also satisfies Hypothesis~\ref{secondorder}, and $d\geq 2$, then
		\begin{equation} \label{eq:musaha2}
		 \mu^\ideal(\beta,\rho) = f_\infty^\cl(\beta) =  e_\infty+ O\Big(\frac{\log \beta}{\beta} \Big)  = \mu(\nu) + O\Big(\frac{\log \beta}{\beta}\Big). 
		\end{equation}
	\end{itemize}
\end{prop}

Next we state the behaviour of $m^\ideal(\beta,\rho) = \sum_{k\in \N} \rho_k^\ideal(\beta,\rho)$, the number of clusters per unit volume. Note that for an ideal mixture, this is essentially the same as the pressure, $\beta p^\ideal(\beta,\rho) = m^\ideal(\beta,\rho)$~\cite{hillbook}. 

\begin{prop}[Number of clusters (pressure)] \label{prop:msaha}
	Suppose that $v$ satisfies Hypotheses~\ref{ass:hoelder} and \ref{ass:maxdist}. 
	Fix $\nu >0$ and put $\rho = \exp(-\beta \nu)$. 
	Then, as $\beta\to\infty$,
	\begin{itemize}
		\item if $\nu \in (\nu^*,\infty)\setminus \mathcal{N}$, 
		$
			m^\ideal(\beta,\rho) = 
				\Bigl(1 + O(\e^{- \beta \Delta(\nu)/2})\Bigr) \rho /k_\nu,
		$
		\item if $\nu < \nu^*$ and in addition $v$ satisfies Hypothesis~\ref{secondorder}, and $d\geq 2$, 
	then $m^\ideal(\beta,\rho) = \exp( -\beta \nu^* + O(\log \beta)) = o(\rho)$. 		 
	\end{itemize} 
\end{prop} 

Finally, we analyse the behaviour of the minimiser of $f^\ideal(\beta,\rho,\cdot)$. 

\begin{prop}[Cluster size distribution] \label{prop:ideal-min}
	Suppose that $v$ satisfies Hypotheses~\ref{ass:hoelder} and~ \ref{ass:maxdist}. 
	Fix $\nu >0$ and put $\rho = \exp(-\beta \nu)$. 
	Then, as $\beta \to \infty$,
	\begin{itemize}
		\item if $\nu\in(\nu^*,\infty)\setminus \mathcal{N}$,
		\begin{equation}
			\label{knucluster}
			\frac{k_\nu \rho_{k_\nu}^\ideal(\beta,\rho)}{\rho} = 1 
				+ O(\e^{-\beta \Delta(\nu)/2}),
			%\qquad\mbox{and}\qquad 
			%\sum_{\heap{k\in \N }{ k\not= k_\nu}}\frac{k_\nu \rho_{k_\nu}^\ideal(\beta,\rho)}{\rho} = O(\e^{-\beta \Delta(\nu)/2}).	
		\end{equation}
		\item if $\nu < \nu^*$ and in addition $v$ satisfies Hypothesis~\ref{secondorder}, and $d\geq 2$, 
		 \begin{equation}\label{Kclusters}
			\sum_{k =1}^{\infty} \frac{k \rho_{k}^\ideal(\beta,\rho)}{\rho} 
				= O( \e^{ - \beta(\nu^*-\nu) +O(\log \beta)})).
		\end{equation}
	\end{itemize}
\end{prop}

The interpretation of \eqref{knucluster} is that all but an exponentially small fraction of particles are in clusters of size $k_\nu$, while the one of \eqref{Kclusters} is that the fraction of particle in 
finite-size clusters  goes to $0$ exponentially fast.

\begin{proof}[Proof of Proposition~\ref{prop:satsaha}] 
	Because of Lemma~\ref{lem:idealtail2}, for suitable $c>0$ and all sufficiently large $K \in \N$ 
	and sufficiently large $\beta$,
	\begin{equation*}
		\sum_{k=1}^K k Z_k^\cl(\beta) \e^{\beta k f_\infty^\cl(\beta)} \leq \rho_\sat^\ideal(\beta) \leq \sum_{k=1}^K k Z_k^\cl(\beta) \e^{\beta k f_\infty^\cl(\beta)} + \rho_\sat^\ideal(\beta) \e^{- \beta c K^{1-1/d}},
	\end{equation*}
	whence we see that
	\begin{equation*}
		\rho_\sat^\ideal(\beta) = \big(1+ O(\e^{- \beta c K^{1-1/d}})\big) \sum_{k=1}^K 
			k Z_k^\cl(\beta) \e^{\beta k f_\infty^\cl(\beta)}.
	\end{equation*}	
	The proof is concluded by choosing $K$ large enough so that 
	 every minimiser of $E_k - k e_\infty$ 
	is smaller or equal to $K$, since for such a $K$, the sum on the right-hand side of the 
	previous equation is $ \exp( - \beta \nu^* + O(\log \beta))$.  
\end{proof}

\begin{proof}[Proof of Proposition~\ref{lem:chempot}]
	Consider first the case $\nu\in(\nu^*,\infty)\setminus \mathcal{N}$. Hence, 
	$\mu(\nu) = (E_{k_\nu}- \nu)/k_\nu$ for a unique $k_\nu \in \N$ and 
	$(E_k - \nu)/k  - \mu(\nu) \geq \Delta(\nu) >0$ for all $k \neq k_\nu$. 	For sufficiently large $\beta$, we will have $\rho < \rho_\sat^\ideal(\beta)$ 
	and therefore the chemical potential is strictly smaller than $f_\infty^\cl(\beta)$ and is given 
	by the unique solution of equation~\eqref{eq:musol} which we rewrite as 	\begin{equation}\label{zpowerseries}
		1 = \sum_{k=1}^\infty k\,z^k\,
			\exp\left(- \beta k \left[  f_k^\cl(\beta)- \frac{\nu}{k}
			-  f_{k_\nu}^\cl(\beta)+ \frac{\nu}{k_\nu} \right]\right)
	\end{equation}
	with the auxiliary variable
	\begin{equation} \label{eq:z}
		z = z(\beta,\rho,\nu) := \exp(\beta \mu^\ideal(\beta,\rho)) \exp\Bigl( - \beta \Bigl[ f_{k_\nu}^\cl(\beta) - \frac{\nu}{k_\nu} \Bigl] \Bigr).
	\end{equation}
	We bound the sum in equation~\eqref{zpowerseries} from below by the summand for $k= k_\nu$. 
	This gives  
	$1 \geq k_{\nu} z^{k_\nu}$ and thus $z \leq 1$. Next, we choose $\beta_0$ 
	such that for all $\beta \geq \beta_0$ and all $k \neq k_\nu$,  the term in square 
	brackets in~\eqref{zpowerseries} is larger than $\beta \Delta(\nu)/2$. 
%	\begin{equation*}
%		f_k^\cl(\beta)+ \frac{\beta^{-1} \log \rho}{k}
%			-  f_{k_\nu}^\cl(\beta)- \frac{\beta^{-1} \log \rho}{k_\nu} 
%		\geq \Delta(\nu)/2.
%	\end{equation*}
	Then
	\begin{equation*}
		1  \leq k_\nu z^{k_\nu} + \sum_{k \neq k_\nu} k e^{- \beta k \Delta(\nu)/2} 
		 \leq k_\nu z^{k_\nu} + \frac{\exp(- \beta \Delta(\nu)/2)}{(1-\exp(- \beta \Delta(\nu)/2))^2}.
	\end{equation*}
	Thus we get $k_\nu z^{k_\nu} = 1+ O(\exp(- \beta \Delta(\nu)/2))$ and~\eqref{eq:musaha} follows from~\eqref{eq:z}. 
%	\begin{align*}
%		\mu^\ideal(\beta,\rho) & = f_{k_\nu}^\cl(\beta) + \frac{\beta^{-1} 
%			\log \rho}{k_\nu} + \beta^{-1} \log \Bigl(\frac{1}{k_\nu} 
%				\Bigl[1+O(e^{- \beta \Delta(\nu)/2})\Bigr]\Bigl)  \\
%			& = f_{k_\nu}^\cl(\beta) + \frac{\beta^{-1} 
%			\log \rho}{k_\nu} - \beta^{-1} \log k_\nu + O(\beta^{-1} e^{-\beta \Delta(\nu)/2}). 
%	\end{align*}
	
	Now let us come to the case $\nu< \nu^*$. Because of Proposition~\ref{prop:satsaha}, 
	for sufficiently large $\beta$, we will have $\rho > \rho_\sat^\ideal(\beta)$ 
	and hence by definition $\mu^\ideal(\beta,\rho) = f_\infty^\cl(\beta)$. 
	Equation~\ref{eq:musaha2} is then a consequence of Lemma~\ref{lem:fk-lowtemp}. 
\end{proof}

\begin{proof}[Proof of Proposition~\ref{prop:msaha}]
	First we consider the case $\nu\in(\nu^*,\infty)\setminus \mathcal{N}$. With $z=z(\beta,\rho,\nu)$ from \eqref{eq:z}, by an argument similar to the proof of Proposition~\ref{lem:chempot}, 
	$m^\ideal(\beta,\rho)/\rho = z^{k_\nu} + O(\exp( -\beta \Delta(\nu))$. 
	Since we saw that $k_\nu z^{k_\nu} = 1 + O(\exp( - \beta \Delta(\nu)))$, we are done. 
	
	For the case $\nu< \nu^*$, we note that for sufficiently large $\beta$, 	$\rho >\rho_\sat^\ideal(\beta)$, hence $m^\ideal(\beta,\rho) = \sum_{k=1}^\infty Z_k^\cl(\beta) \exp( - \beta k f_k^\cl(\beta))$ and the claim follows by an argument similar to the 
	proof of Prop.~\ref{prop:satsaha}.   
\end{proof} 

\begin{proof} [Proof of Proposition~\ref{prop:ideal-min}]
The case $\nu\in(\nu^*,\infty)\setminus \mathcal{N}$ is a consequence of the identity  $k _\nu \rho_{k_\nu}^\ideal(\beta,\rho) /\rho = k_\nu z^{k_\nu}$ 
	and the argument in the proof of Proposition~\ref{lem:chempot}. 
	
In the case $\nu < \nu^*$ we just remark that for sufficiently large $\beta$, 
	$\rho > \rho_\sat^\ideal(\beta)$ hence 
	\begin{equation*}
		\sum_{k=1}^\infty \frac{k\rho_k^\ideal(\beta,\rho)}{\rho} 
		=\frac{\rho_\sat^\ideal(\beta)}{\rho} 
	\end{equation*}
	and the proof is concluded by applying Proposition~\ref{prop:satsaha}. 
\end{proof}

\medskip

\noindent \textbf{Acknowledgements}
We gratefully acknowledge financial support by the
DFG-Forscher\-gruppe~FOR718 ``Analysis and stochastics in complex physical systems''.\\

\end{document}